\documentclass[12pt]{amsart}
\usepackage{txfonts}
\usepackage{amssymb}
\usepackage{eucal}
\usepackage{graphicx}
\usepackage{amsmath}
\usepackage{amscd}
\usepackage[all]{xy}           

\usepackage{amsfonts,latexsym}
\usepackage{xspace}
\usepackage{epsfig}
\usepackage{float}
\usepackage{color}
\usepackage{fancybox}
\usepackage{colordvi}
\usepackage{multicol}
\usepackage{colordvi}
\usepackage{stmaryrd}

\newtheorem{theorem}{Theorem}[section]
\newtheorem{prop}[theorem]{Proposition}
\newtheorem{defn}[theorem]{Definition}
\newtheorem{lemma}[theorem]{Lemma}
\newtheorem{coro}[theorem]{Corollary}
\newtheorem{prop-def}{Proposition-Definition}[section]

\newcommand{\nc}{\newcommand}
\newcommand{\delete}[1]{}

\nc{\mlabel}[1]{\label{#1}}  
\nc{\mcite}[1]{\cite{#1}}  
\nc{\mref}[1]{\ref{#1}}  
\nc{\mbibitem}[1]{\bibitem{#1}} 

\delete{
\nc{\mlabel}[1]{\label{#1}  
{\hfill \hspace{1cm}{\bf{{\ }\hfill(#1)}}}}
\nc{\mcite}[1]{\cite{#1}{{\bf{{\ }(#1)}}}}  
\nc{\mref}[1]{\ref{#1}{{\bf{{\ }(#1)}}}}  
\nc{\mbibitem}[1]{\bibitem[\bf #1]{#1}} 
}

\nc{\bfk}{\mathbf{k}}
\nc{\Der}{\mathrm{Der}}
\nc{\Ker}{\mathrm{Ker}}

\begin{document}

\title{Semi-Associative $3$-Algebras }\footnotetext{ Corresponding author: Ruipu Bai, E-mail: bairuipu@hbu.edu.cn.}

\author{RuiPu  Bai}
\address{College of Mathematics and Information Science,
Hebei University}
\email{bairuipu@hbu.edu.cn}

\author{Yan Zhang}
\address{College of Mathematics and Information  Science,
Hebei University, Baoding 071002, China} \email{zhycn0913@163.com}

\date{}

\maketitle


\renewcommand{\thefootnote}{}

\footnote{2010 \emph{Mathematics Subject Classification}: 15A69; 17D99.}

\footnote{\emph{Key words and phrases}: semi-associative 3-algebra, 3-Lie algebra, double module, extension}

\renewcommand{\thefootnote}{\arabic{footnote}}
\setcounter{footnote}{0}


\begin{abstract}

 A new 3-ary non-associative algebra, which is called a semi-associative $3$-algebra, is introduced, and the  double modules and double   extensions by  cocycles  are provided. Every semi-associative $3$-algebra $(A, \{ , , \})$ has an adjacent 3-Lie algebra  $(A, [ , , ]_c)$.  From a semi-associative $3$-algebra $(A,  \{, , \})$, a double module $(\phi, \psi, M)$ and a cocycle $\theta$, a semi-direct product semi-associative $3$-algebra $A\ltimes_{\phi\psi} M $ and a double extension $(A\dot+A^*, \{ , , \}_{\theta})$ are  constructed,  and structures are studied.
\end{abstract}

\section{Introduction}

 3-ary algebraic systems are applied in
mathematics and mathematical physics. For example, 3-Lie algebras
provided by Filippov  in 1985 \cite{FV, B1}, have close relation with
the theory of integrable-systems and Nambu mechanical-systems
\cite{BL, BW}.  Bagger and Lambert \cite{BW2,BW3, BW4} proposed a
supersymmetric fields theory model for multiple $M_2$-branes based
on the metric $3$-Lie algebras \cite{BW7}.

C. Bai and coauthors \cite{BGS}, defined pre-3-Lie algebras and local cocycle 3-Lie bialgebras and generalized Yang-Baxter equation in 3-Lie algebras, and proved that  tensor form solutions  of 3-Lie Yang-Baxter equation can be constructed by   local cocycle 3-Lie bialgebras. Three classes of 3-algebras are constructed in \cite{CFJMZ}:

(1) The classical Nambu bracket based on three variables ( say $x, y, z$ ) is the simplest to
compute, in most situations. It involves the Jacobian-like determinant of partial derivatives
of three functions $A, B, C$, $$ \{ A, B, C\}=\frac{\partial(A, B, C)}{(\partial x, \partial y, \partial z)},$$
and it satisfies Filippove identity:
$$
\{A, B, \{C, D, E\}\}=\{\{A, B,C\}, D, E\}+\{C, \{A, B, D\}, E\}+ \{C, D, \{A, B, E\}\}.
$$

(2) The  sum of single operators producted with commutators of the remaining two,
or as anticommutators acting on the commutators,

$$[A, B, C ] = A [B, C] + B [C, A] + C [A, B] = [B, C] A + [C, A] B + [A, B] C .$$

The multiplication does not satisfy Filippov identity, and the 3-algebra can by
realized by operators  \cite{CFJMZ}.

(3) The re-packaging
commutators of any Lie algebra to define

$$\langle A, B, C\rangle = [A, B] T r (C ) + [C, A] T r (B) + [B, C] T r (A). $$
\\
This is again a totally antisymmetric trilinear combination, but it is a singular construction
for finite dimensional realizations in the sense that $T r \langle A, B, C\rangle = 0,$ and authors in \cite{BBW3} proved that all the $m$-dimensional 3-Lie algebras with $m\leq 5$,  except for the unique  simple 3-Lie algebra,
can be constructed by the method from Lie algebras.

Likewise, in \cite{BGS}, the paper provides an algebra which is called the $3$-pre-Lie algebra, and induces a sub-adjacent $3$-Lie algebra. However,
 the $3$-pre-Lie algebra does not satisfy associative property. Motivated by this, in this paper, we define a new 3-algebra, the 3-ary multiplication is none completely antisymmetry, but it has close relation with 3-Lie algebras.

In the following we assume that $\mathbb Z$ is the set of integers,  $\mathbb F$ is an algebraic closed field of characteristic zero, and for a subset $S$ of a linear space $V$,  $\langle S\rangle $ denotes the subspace spanned by $S$.
And we usually omit  zero products of basis vectors when we list  the multiplication of $3$-algebras.

\section{ Natures of semi-associative 3-algebras}

\begin{defn}  A semi-associative 3-algebra $(A, \{ , , \})$ is a vector space $A$ with a 3-ary linear multiplication $\{
, , \}: A\otimes A\otimes A\rightarrow A$ satisfying that  for all
$x_i\in A, 1\leq i\leq 5$,
\begin{equation}
\{x_1, x_2, x_3\}=-\{x_2, x_1, x_3\},
\label{eq:21}
\end{equation}
\begin{equation}
\{x_1, \{x_2, x_3, x_4\}, x_5\}=\{x_1, x_2, \{x_3, x_4, x_5\}\},
\label{eq:22}
\end{equation}
\begin{equation}
\{x_1, \{x_2, x_3, x_4\}, x_5\}=\{x_5, \{x_2, x_3, x_4\}, x_1\}+\{x_1, \{x_5,
x_3, x_4\}, x_2\}.
\label{eq:23}
\end{equation}

\label{defn:21}
\end{defn}

For subspaces $B_1, B_2, B_3$ of $A$,
denotes $\{B_1, B_2, B_3\}$  the subspace spanned by  vectors $\{x_1, x_2, x_3\}$, $\forall x_i\in B_i$, $i=1, 2, 3$;  $\{A, A, A\}$ is denoted by $A^1$, which is called {\it the derived algebra} of $A$. If  $A^1=0$, then $A$ is called abelian.

For example, let $A$ be a $3$-dimensional vector space with   a basis $v_1, v_2, v_3$. Then $A$ is a semi-associative 3-algebra with the multiplication
$$\{v_1, v_2, v_2\}=-\{v_2, v_1, v_2\}=v_3,~~ \mbox{ and others are zero.}$$

\begin{defn} A subalgebra  of a semi-associative 3-algebra $A$ is a subspace $B$ satisfying $\{B, B, B\}\subseteq B$. If $B$ satisfies that $\{A, A, B\}\subseteq B$ and $(A, B, A)\subseteq B$, then $B$ is called an ideal of $A$.

\label{defn:22}
\end{defn}

Obviously, $\{0\}$ and $A$ are ideals of $A$, which are called trivial ideals. If $A$ has non proper ideals, then $A$ is called {\it  a simple semi-associative 3-algebra.}

For a given subspace $V\leq A$, the subalgebra

$$Z_A(V)=\{~ x ~|~ x\in A, ~\{x, V, A~\}=\{~V, A, x~\}=0~\}, $$
is called {\it the centralizer of $V$ in $A$}.  $Z_A(A)$ is called {\it the center} of $A$, and is simply denoted by $Z(A)$, that is,
$Z(A)=\{~ x  \mid   x\in A, \{x, A, A\}=\{A, A, x\}=0~\}$. It is clear that  $Z(A)$ is an ideal.

\begin{defn} Let $A$ and $A_1$ be semi-associative 3-algebras. If a linear mapping (linear isomorphism) $f: A\rightarrow A_1$
satisfies
$$
f\{x, y, z\}=\{f(x), f(y), f(z)\},~~ \forall x, y, z\in A,$$
then $f$ is called an algebra homomorphism (an algebra isomorphism).

\label{defn:21}
\end{defn}

Let $I$ be an ideal of $A$. Then quotient space $A/I=\{ x+I \mid ~~x\in A\}$ is a semi-associative 3-algebra in the multiplication
 $$\{x+I, y+I, z+I\}=\{x, y, z\}+I, ~~ \forall x, y, z\in A,$$ which is called {\it the quotient algebra of $A$ by $I$}, and
 $\pi: A\rightarrow A/I$, $\pi(x)=x+I$,  $\forall x\in A$, is a surjective algebra homomorphism.

\begin{prop} Let $A$  be a semi-associative algebra, $I_1, I_2, I_3$ be ideals of $A$. Then

1) $I_1+I_2$, $I_1\cap I_2$ and $\{I_1, I_2, I_3\}$ are ideals of $A$.

2) If $I_1\subseteq I_2$, then $I_2/I_1$ is an ideal of $A/I_1$.
\label{prop:21}
\end{prop}

\begin{proof} It  can be verified by a direct computation according to Definition \ref{defn:22}.
\end{proof}

\begin{prop} Let $A$  and $A_1$ be  semi-associative algebras, $f: A\rightarrow A_1$ be an algebra  homomorphism.  Then

1) $K=Ker f=\{x\in A\mid f(x)=0\}$ is an ideal of $A$, and $f(A)$ is a subalgebra of $A_1$.

2) If $f(A)=A_1$, then linear  mapping $\bar f: A/K\rightarrow A_1$, $\bar f(x+K)=f(x)$, for all $x\in A$, is an algebra isomorphism, and there is one to one correspondence  between  subalgebras of $A$ containing  $K$, with the subalgebras of $A_1$, and ideals  correspond to ideals.

\label{prop:22}
\end{prop}

\begin{proof} The result is easily  verified by a direct computation.
\end{proof}

\begin{theorem}
 Let $A$ be a semi-associative 3-algebra.

1) If there exist nonzero vectors $e_1, e_2, e_3, e_4\in A$ such that $\{e_1, e_2, e_3\}=e_4$, then $e_4\neq \lambda e_s$, where $\lambda\in \mathbb F, \lambda\neq 0, $  $s=1, 2, 3$.

2) If there exist nonzero vectors  $e_1, e_2, e_3\in A$ such that $\{e_1, e_2, e_3\}\neq 0$, then $\{e_1, e_2, e_3\}\neq \lambda e_1+ \mu e_2, \forall \lambda, \mu\in \mathbb F$.  Therefore,  if
 $\{e_1, e_2, e_3\}=ae_1 + be_2+ce_4\neq 0$, then $c\neq 0$, and $e_1, e_2, e_4$ are linearly independent.

\label{thm:23}
\end{theorem}

\begin{proof}~ If there exist nonzero vectors $e_1, e_2, e_3, e_4\in A$ such that $\{e_1, e_2, e_3\}=e_4=\lambda e_1$, $\lambda\in \mathbb F, \lambda\neq 0$, then by Eqs \eqref{eq:21} and \eqref{eq:22}

$e_4=\{e_1, e_2, e_3\}=\frac{1}{\lambda}\{e_2, \{e_2, e_1, e_3\},
e_3\}=\frac{1}{\lambda}\{e_2, e_2, \{e_1, e_3,
e_3\}\}=0$.
Contradiction. Therefore, $e_4\neq \lambda e_1$.
By the similar discussion to the above,  $e_4 \neq \mu e_2$, $\mu\in \mathbb F$.

If $e_4=\alpha e_3$, $\alpha\in \mathbb F, \alpha\neq 0$,  then

$e_4=\{e_1, e_2, e_3\}=\frac{1}{\alpha}\{e_1, e_2, \{e_1, e_2,
e_3\}\}=-\frac{1}{\alpha}\{e_1, e_1, \{e_2, e_2,
e_3\}\}=0$.
Contradiction. The result 1) follows.

If there exist nonzero vectors  $e_1, e_2, e_3\in A$ such that $\{e_1, e_2, e_3\}=\lambda e_1+ \mu e_2\neq 0, \lambda, \mu\in \mathbb F$.
 Without loss of generality, we can suppose $\mu\neq 0$. Then
 $$\{e_1, \lambda e_1+ \mu e_2, \frac{1}{\mu}e_3\}=\lambda e_1+ \mu e_2\neq 0.$$
 Contradiction. Therefore, $\{e_1, e_2, e_3\}\neq \lambda e_1+ \mu e_2, \forall \lambda, \mu\in \mathbb F$.
The proof is complete.
\end{proof}

\begin{theorem} Let $A$ be a non-abelian semi-associative 3-algebra with $\dim A=m\geq 3$. Then there exist linearly independent vectors $e_i, e_j, e_k\in A$ such that
$$\{e_i, e_j, e_k\}\neq 0.$$
\label{thm:24}
\end{theorem}

\begin{proof} Since $A$ is non-abelian, by \eqref{eq:21}, there are linearly independent vectors $e_1, e_2\in A$, such that $\{e_1, e_2, A\}\neq 0.$
Suppose  $\{ e_1, e_2, \cdots, e_m\} $ is a basis of $A$.  If $\{e_1, e_2, e_l\}=0$, for all $l\geq 3$, then there exist $a, b\in \mathbb F$, such that
$\{e_1, e_2, ae_1+be_2\}\neq 0$, we get  $\{e_1, e_2, ae_1+be_2+e_3\}\neq 0$, and vectors $e_1, e_2, ae_1+be_2+e_3$ are linearly independent. The proof is complete.
\end{proof}

\begin{theorem} Let $A$ be an $s$-dimensional semi-associative 3-algebra with $s\leq 6$. Then $A^{1}\subseteq Z(A)$.
\label{thm:25}
\end{theorem}

\begin{proof} If $A$ is abelian, then the result is trivial.

If $A^1\neq 0$, then we will  prove  $\{A, A^1, A\}=\{A, A, A^1 \}=0$.

Since $A^1\neq 0$, then  there are  nonzero vectors $e_i, e_j, e_k, e_l\in A$ such that $$\{e_i, e_j, e_k\}=e_l.$$

$ \bullet$   $\{ e_i, e_j, e_k \}$ are linearly dependent.

Then we can  suppose $e_k=ae_i+be_j$, $a, b\in \mathbb F$, and
   $e_l=a\{e_i, e_j, e_i\}+b\{e_i, e_j, e_j\}$.

  By \eqref{eq:22} and \eqref{eq:23},  for all $e_m, e_n\in A$, we have

 $\{e_m, e_l, e_n\}=a\{e_m, \{e_i, e_j, e_i\}, e_n\}+b\{e_m, \{e_i, e_j, e_j\}, e_n\}$

 \hspace{2cm}$=-a\{e_m, e_j, \{e_i, e_i, e_n\}\}+b\{e_m, e_i, \{e_j, e_j, e_n\}\}=0$,

$\{e_m, e_n, e_l\}=a\{e_m, e_n, \{e_i, e_j, e_i\}\}+b\{e_m, e_n, \{e_i, e_j, e_j\}\}$

\hspace{2cm}$=a\{e_i, e_i, \{e_n, e_j, e_m\}\}+b\{e_j, e_j, \{e_n, e_i, e_m\}=0.$
\\Therefore, $\{A, e_l, A\}=\{A, A, e_l\}=0$.

$\bullet\bullet$  $\{ e_i, e_j, e_k \}$ are linearly independent.

$\bullet\bullet_1$ If there $a, b, c\in \mathbb F,$ such that  $e_l=ae_i+be_j+ce_k$, then
$$\{e_i, e_j, e_k\}=ae_i+be_j+ce_k, ~~ a, b, c\in \mathbb F. $$
 By Theorem  \ref{thm:23}, $ac\neq 0$, or $bc\neq 0$. Without loss of generality, suppose $ac\neq 0$,  then
 $$\frac{c}{a}\{e_j, e_k, e_k\}=\frac{1}{a}\{e_j, ce_k, e_k\}=\frac{1}{a}\{e_j, \{e_i, e_j, e_k\}-ae_i-be_j, e_k\}=e_l. $$
 Therefore, for all $e_m, e_n\in A$,

   $\{e_m, e_l, e_n\}=\{e_m, \frac{c}{a}\{e_j, e_k, e_k\}, e_n\}=\{e_m, \frac{c}{a}e_j, \{e_k, e_k, e_n\}\}=0$, and

   $\{e_m, e_n, e_l\}=\{e_m, e_n, \frac{c}{a}\{e_j, e_k, e_k\}\}= \frac{c}{a}\{e_k, \{e_k, e_n, e_j\}, e_m\}=0$.
 \\ It follows  $\{A, A, e_l\}=\{A, e_l, A\}=0.$

 $\bullet\bullet_2$  $\{ e_i, e_j, e_k, e_l \}$ are linearly independent.

Then in the case $\dim A=4$, $\{ e_i, e_j, e_k, e_l\}$ is a
basis of  $A$.  For all $e_m, e_n\in A$, set

 $e_m=a_1e_i+b_1e_j+c_1e_k+d_1e_l, $  $e_n=a_2e_i+b_2e_j+c_2e_k+d_2e_l,
$ $a_i, b_i, c_i, d_i\in \mathbb F, i=1, 2$.
Then $\{e_m, e_l, e_n\}=\{a_1e_i+b_1e_j+c_1e_k+d_1e_l, e_l, e_n\}$

\hspace{2.6cm}$=a_1\{e_i, e_i, \{e_j, e_k, e_n\}\}-b_1\{e_j, e_j, \{e_i, e_k, e_n\}\}-c_1\{e_k, e_i, \{e_k, e_j, e_n\}\}$

\hspace{2.6cm}$=c_1\{e_k, e_k, \{e_i, e_j, e_n\}\}=0.$\\
By the similar to the above, $\{e_m, e_n, e_l\}=0$.

In the case  $\dim A=5$, then we can suppose that $\{ e_i, e_j, e_k, e_l, e_t\}$ is a  basis of $A$.

From the discussion of the case $\dim A=4$, we only  need to prove that $\{e_t, e_l, e_t\}=0$.
 Thanks to \eqref{eq:21}, \eqref{eq:22} and \eqref{eq:23},

$\{e_t, e_l, e_t\}=\{e_t, \{e_i, e_j, e_k\}, e_t\}=-\{e_t, e_t, \{e_j, e_k, e_i\}\}=0$.

In the case $\dim A=6$, then we can suppose $e_1, e_2, e_3, e_4, e_5, e_6$ is  a basis of $A$, where $e_1=e_i$, $e_2=e_j$, $e_3=e_k$, and $e_4=e_l$. Then
$\{e_1, e_2, e_3\}=e_4.$

By the above discussion,

\begin{equation}
\{B, e_4, B\}=\{B, B, e_4\}=0,~~ \{C, e_4, C\}=\{C, C, e_4\}=0,
\label{eq:24}
\end{equation}
\\
where $B=\langle e_1, e_2, e_3, e_4, e_5\rangle$, $C=\langle e_1, e_2, e_3, e_4, e_6\rangle$. Therefore, we  only need to discuss the products
$\{e_5, e_4, e_6\}$,
$ \{e_6, e_4, e_5\}$ and $\{e_5, e_6, e_4\}$. Suppose
$$\{e_5, e_4, e_6\}=a_3e_1+b_3e_2+c_3e_3+d_3e_4+\lambda_3 e_5+\mu_3 e_6, ~~ a_3, b_3, c_3, \lambda_3, \mu_3\in \mathbb F. $$
Thanks to \eqref{eq:24},

$\lambda_3\{e_5, e_4, e_6\}=\{\{e_5, e_4, e_6\}-a_3e_1-b_3e_2-c_3e_3-d_3e_4-\mu_3 e_6, e_4, e_6\}=\{\{e_5, e_4, e_6\}, e_4, e_6\}$

\hspace{2.35cm}$=\{e_4, \{e_4, e_5, e_6\}, e_6\}=\{e_4, e_4, \{e_5, e_6, e_6\}\}=0.$
\\
Therefore, $\lambda_3\{e_5, e_4, e_6\}=\lambda_3 ae_1+\lambda_3 b_3e_2+\lambda_3c_3e_3+\lambda_3 d_3e_4+\lambda_3^2 e_5+\lambda_3 \mu-3\lambda_3 e_6=0$,
 $\lambda_3=0$.

  By the complete similar discussion,  we have $a_3=b_3=c_3=d_3=\mu_3=0$.

  Therefore,
$\{ e_5, e_4, e_6\}=\{ e_5, e_6, e_4\}$  $=\{e_6, e_4, e_5\}=0$, and
  $A^{1}\subseteq Z(A)$.
\end{proof}

\section{ Derivations and centroid of semi-associative 3-algebras}

\subsection{ Derivations of semi-associative 3-algebras}

\begin{defn}
 Let $A$ be a semi-associative 3-algebra, and $D: A\rightarrow A$ be a linear mapping. If $D$ satisfies that
\begin{equation}
D\{x_1, x_2, x_3\}=\{Dx_1, x_2, x_3\}+\{x_1, Dx_2, x_3\}+\{x_1, x_2, Dx_3\}, ~~\forall  x_1, x_2, x_3 \in A,
\label{eq:31}
\end{equation}
then $D$ is called a derivation of $A$. $Der(A)$ denotes the set of all derivations  of  $A$.
\label{defn:31}
\end{defn}

\begin{prop}
 $Der(A)$ is a subalgebra of the general linear algebra $gl(A)$.
 \label{prop:31}
 \end{prop}

 \begin{proof}
 The result follows from a direct computation.
 \end{proof}

For $D\in Der(A)$, if $D$ satisfies $D(A)\subseteq Z(A)$ and $D(A^1)=0$, then $D$ is  called a {\it central derivation} of $A$.
And $Der_C(A)$ denotes the set of  central derivations. It is clear that $Der_C(A)$ is a subalgebra of $Der(A)$.

 For all $x_1, x_2\in A$, define linear mappings,

 $L(x_1, x_2),$ ~$ R(x_1, x_2), ~~ S(x_1, x_2): A \times A \longrightarrow A$,  by
\begin{equation}
L(x_1, x_2)(x)=\{x_1, x_2, x\},~ R(x_1, x_2)(x)=\{x, x_1, x_2\}, ~~ \forall x\in A,
\label{eq:32}
\end{equation}
\begin{equation}
 S(x_1, x_2)=L(x_1, x_2)-R(x_1, x_2).
 \label{eq:33}
\end{equation}
 $L(x_1, x_2)$  and $ R(x_1, x_2)$ are called the left multiplication and the right multiplication, respectively.

\begin{lemma} Let $A$ be a semi-associative 3-algebra. Then  for all $x_1, x_2, x_3, x_4\in A$,
\begin{equation}
L(x_1, x_2)=-L(x_2, x_1),
\label{eq:34}
\end{equation}
\begin{equation}L(x_1, x_2)L(x_3, x_4)=L(x_1, \{x_2, x_3, x_4\}),
\label{eq:35}
\end{equation}
\begin{equation}L(x_1, \{x_2, x_3, x_4\})=R(\{x_2, x_3, x_4\}, x_1)-R(x_1, x_2)R(x_3, x_4),
\label{eq:36}
\end{equation}
\begin{equation}
L(x_1, x_2)L(x_3, x_4)=L(x_3, x_1)L(x_2, x_4)=L(x_4, x_2)R(x_3, x_1)+L(x_3, x_1)R(x_4, x_2),
\label{eq:37}
\end{equation}
\begin{equation}
R(x_3,x_4)(R(x_1,x_2)+R(x_2, x_1))=0.
\label{eq:38}
\end{equation}
\label{lem:31}
\end{lemma}

\begin{proof} Eqs \eqref{eq:34}, \eqref{eq:35} and  \eqref{eq:36}  follow from \eqref{eq:21}, \eqref{eq:22} and \eqref{eq:23}, directly.

  For $\forall x_i\in A$, $1\leq i\leq 5$, by \eqref{eq:22} and \eqref{eq:23}, we have

$R(x_3, x_4)R(x_1, x_2)(x_5)+R(x_3, x_4)R(x_2, x_1)(x_5)$

\vspace{2mm}\hspace{-3mm}$ =\{\{x_5, x_1, x_2\}, x_3, x_4\}+\{\{x_5, x_2, x_1\}, x_3, x_4\}$

\vspace{2mm}\hspace{-3mm}
$ =-\{x_3, \{x_5, x_1, x_2\}, x_4\}-\{x_3, \{x_5, x_2, x_1\}, x_4\}$

\vspace{2mm}\hspace{-3mm}$ =-\{x_3, \{x_5, x_1, x_2\}, x_4\}+\{x_3, \{x_5, x_1, x_2\}, x_4\} =0,$

Similarly, we have

 $L(x_1, x_2)L(x_3, x_4)(x_5)-L(x_3, x_1)L(x_2, x_4)(x_5)$

\vspace{2mm}\hspace{-3mm}$=\{x_1, x_2, \{x_3, x_4, x_5\}\}-\{x_3, x_1, \{x_2, x_4, x_5\}\} =0,$ and

$L(x_1, x_2)L(x_3, x_4)(x_5)-(L(x_4, x_2)R(x_3, x_1)(x_5)+L(x_3, x_1)R(x_4, x_2)(x_5))=0.$

It follows \eqref{eq:37} and \eqref{eq:38}.
\end{proof}

Let $L(A)$, $R(A)$ and $S(A)$ be subspaces of $End(A)$ spanned by linear mappings $L(x_1, x_2)$, $R(x, x_2)$ and $S(x_1, x_2)$, respectively, $\forall x_1, x_2\in A,$ that is,

 $$L(A)=\langle L(x_1, x_2)\mid \forall x_1, x_2\in A\rangle,~~R(A)=\langle R(x, x_2)\mid \forall x_1, x_2\in A\rangle,$$

$$S(A)=\langle S(x_1, x_2)\mid \forall x_1, x_2\in A\rangle, ~~ T(A)=L(A)+R(A).$$

 Then  $L(A)\subseteq T(A)$ and $R(A)\subseteq T(A).$

\begin{theorem} Let $A$  be a semi-associative 3-algebra. Then $T(A)$ is an abelian subalgebra of $gl(A)$.
\label{thm:31}
\end{theorem}

\begin{proof} By \eqref{eq:21} and \eqref{eq:22},  for all $x_i\in A$, $1\leq i\leq 5,$

$
[L(x_1, x_2), L(x_3, x_4)](x_5)=(L(x_1, x_2)L(x_3, x_4)-L(x_3, x_4)L(x_1, x_2))(x_5)$

\vspace{2mm}\hspace{-3mm}$=\{x_1, \{x_2, x_3, x_4\}, x_5\}- \{x_3, \{x_4, x_1, x_2\}, x_5\}
$

\vspace{2mm}\hspace{-3mm}$=\{x_1, x_3, \{x_4, x_2, x_5\}\}- \{x_3, \{x_4, x_1, x_2\}, x_5\}
$

\vspace{2mm}\hspace{-3mm}$=(\{x_3, \{x_4, x_1, x_2\}, x_5\}- \{x_3, \{x_4, x_1, x_2\}, x_5\}=0.
$

Similarly, $
[L(x_1, x_2), R(x_3, x_4)]x_5=0.$ Thanks  to \eqref{eq:23},

$
[R(x_1, x_2), R(x_3, x_4)](x_5)=(R(x_1, x_2)R(x_3, x_4)-R(x_3, x_4)R(x_1, x_2))(x_5)$

\vspace{2mm}\hspace{-3mm}$=\{\{x_5, x_3, x_4\}, x_1, x_2\}-\{\{x_5, x_1, x_2\}, x_3, x_4\}
$

\vspace{2mm}\hspace{-3mm}$=-\{x_1, \{x_5, x_3, x_4\}, x_2\}+\{x_3, \{x_5, x_1, x_2\}, x_4\}
$

\vspace{2mm}\hspace{-3mm}$=-\{x_1, \{x_5, x_3, x_4\}, x_2\}+\{x_5, \{x_1, x_3, x_2\}, x_4\}
$

\vspace{2mm}\hspace{-3mm}$=-\{x_2, \{x_5, x_3, x_4\}, x_1\}-\{x_1, \{x_2, x_3, x_4\}, x_5\}
$
$+\{x_4, \{x_1, x_3, x_2\}, x_5\}$

\vspace{2mm}\hspace{-3mm}$+\{x_5, \{x_4, x_3, x_2\}, x_1\}=0. $

Therefore,
$[L(A), T(A)]=[R(A), T(A)]=0, [T(A), T(A)]=0.$ The proof is complete.
\end{proof}

\begin{theorem} Let $A$ be a semi-associative 3-algebra. Then

1) for any $ x_1, x_2\in A$,  $S(x_1, x_2)$ is a derivation of $A$;

2) $S(A)$ is an ideal of $Der(A)$, and $[S(A), S(A)]=0$.

\label{thm:32}
\end{theorem}

\begin{proof} For $\forall x_1, x_2, x, y, z\in A$, by \eqref{eq:32} and \eqref{eq:33},

$S(x_1, x_2)\{x, y, z\}$$=\{x_1, x_2, \{x, y, z\}\}-\{\{x, y, z\}, x_1,
x_2\}$

\vspace{2mm}\hspace{-3mm}
$=\{x_1, \{x_2, x, y\}, z\}+\{x_1, \{x, y, z\}, x_2\}, $

\vspace{2mm}$\{S(x_1, x_2)(x), y, z\}+\{x, S(x_1, x_2)(y), z\}+\{x, y, S(x_1, x_2)(z)\}$

\vspace{2mm}\hspace{-3mm}
$=\{\{x_1, x_2, x\}, y, z\}-\{\{x, x_1, x_2\}, y, z\}+\{x, \{x_1, x_2, y\}, z\}$

\vspace{2mm}\hspace{-3mm}
$-\{x, \{y, x_1, x_2\}, z\}+\{x, y, \{x_1, x_2, z\}\}-\{x, y, \{z, x_1, x_2\}\}$

\vspace{2mm}\hspace{-3mm}
$=\{x_1, \{x_2, x, y\}, z\}+\{x_1, \{x, y, z\}, x_2\}. $
Therefore,

$S(x_1, x_2)\{x, y, z\}=\{S(x_1, x_2)(x), y, z\}+\{x, S(x_1, x_2)(y), z\}+\{x, y, S(x_1, x_2)(z)\}.$
\\The result 1) follows.

For any $S(x_1, x_2)\in S(A), D\in Der(A)$, and $x\in A$, since

$[S(x_1, x_2), D ](x)$
$=S(x_1, x_2)D(x)-DS(x_1, x_2)(x)$

\vspace{2mm}\hspace{-3mm}
$=\{x_1, x_2, D(x)\}-\{D(x), x_1, x_2\}-D\{x_1, x_2, x\}+D\{x, x_1, x_2\}$

\vspace{2mm}\hspace{-3mm}
$=\{x_1, x_2, D(x)\}-\{D(x),x_1, x_2\}-\{D(x_1), x_2, x\}-\{x_1, D(x_2), x\}$

\vspace{2mm}\hspace{-3mm}
$-\{x_1, x_2, D(x)\}+\{D(x),x_1, x_2\}+\{x, D(x_1), x_2\}+\{x, x_1, D(x_2)\}$

\vspace{2mm}\hspace{-3mm}
$=(S\{x_1, D(x_2)\}-S\{D(x_1), x_2\})(x)$.
\\ Then $[S(A), Der(A)]\subseteq S(A)$, that is, $S(A)$ is an ideal of $Der(A)$.
Thanks to Theorem \ref{thm:31}, $[S(A), S(A)]=[L(A)-R(A), L(A)-R(A)]=0. $
It follows 2).
\end{proof}

For any $x_1, x_2\in A$, $S(x_1, x_2) $  is called an inner derivation of $A$.

\subsection{The centroid  of  semi-associative 3-algebras}

\begin{defn}~ Let $A$ be a semi-associative 3-algebra. The vector space
$$
\Gamma(A)=\{\varphi \in End(A)| \varphi\{x_1, x_2, x_3\}=\{\varphi(x_1), x_2, x_3\}=\{x_1, x_2, \varphi(x_3)\}, \forall x_1, x_2, x_3\in A \}
$$
is called the centroid of $A$.
\label{defn:61}
\end{defn}

For any $\varphi\in End(A)$, by \eqref{eq:21},  $\varphi\in \Gamma(A) $ if and only if for all $x_1, x_2, x_3\in A$,

$$\varphi\{x_1, x_2, x_3\}=\{\varphi(x_1), x_2, x_3\}=\{x_1, \varphi(x_2), x_3\}=\{x_1, x_2, \varphi(x_3)\}.$$

\begin{theorem}~ Let $A$ be a semi-associative 3-algebra. Then

1) $\Gamma(A)$ is a subalgebra of the general linear Lie algebra $gl(A)$.

2) For any $\varphi\in \Gamma(A)$, if $\varphi(A)\subseteq Z(A)$ and $\varphi(A^1)=0$, then $\varphi$ is a central derivation.

3) For any $\varphi\in\Gamma(A)$, and $D\in Der(A)$. Then $\varphi D\in Der(A)$.

4) $Der_C(A)=\Gamma(A)\cap Der(A)$.

\label{thm:61}
\end{theorem}

\begin{proof} For all $\varphi_1, \varphi_2\in \Gamma(A)$, and  $x_1, x_2, x_3\in A$, by Definition \ref{defn:61},

$[\varphi_1, \varphi_2]\{x_1, x_2, x_3\}=(\varphi_1\varphi_2-\varphi_2\varphi_1)\{x_1, x_2, x_3\}$

\vspace{2mm}\hspace{-3mm}$=\{\varphi_1\varphi_2(x_1), x_2, x_3\}-\{\varphi_2\varphi_1(x_1), x_2, x_3\}$
$=\{[\varphi_1, \varphi_2](x_1), x_2, x_3\}$,

\vspace{2mm}$[\varphi_1, \varphi_2]\{x_1, x_2, x_3\}=(\varphi_1\varphi_2-\varphi_2\varphi_1)\{x_1, x_2, x_3\}$

\vspace{2mm}\hspace{-3mm}$=\{x_1, x_2, \varphi_1\varphi_2(x_3)\}-\{x_1, x_2, \varphi_2\varphi_1(x_3)\}$
$=\{x_1, x_2, [\varphi_1, \varphi_2](x_3)\}$.
\\
Therefore,  $[\varphi_1, \varphi_2]\in gl(A)$, the result 1) follows.

 For any $\varphi\in \Gamma(A)$, if $\varphi(A)\subseteq Z(A)$ and $\varphi(A^1)=0$, then by Definition \ref{defn:61} and \eqref{eq:31}, $\varphi\in Der A$. It follows 2).

For any  $\varphi\in \Gamma(A)$, $D\in Der A$, and $x_1, x_2, x_3\in A$,

 $\varphi D\{x_1, x_2, x_3\}=\varphi\{D(x_1), x_2, x_3\}+\varphi\{x_1, D(x_2), x_3\}+\varphi\{x_1, x_2, D(x_3)\}$

~~~~~~~~~~~~~~~~~~~$=\{\varphi D(x_1), x_2, x_3\}+\{x_1, \varphi D(x_2), x_3\}+\{x_1, x_2, \varphi D(x_3)\}$,
\\ $\varphi D\in Der A$.  It follows  3).

For any $\varphi\in \Gamma(A)\cap Der(A)$, by Definition \ref{defn:61}, and \eqref{eq:31},
$\forall x_1, x_2, x_3 \in A$,  we have

$\varphi\{x_1, x_2, x_3\}=\{\varphi(x_1), x_2,  x_3\}+\{x_1,
\varphi(x_2), x_3\}+\{x_1, x_2, \varphi(x_3)\}=3\varphi\{x_1, x_2, x_3\}$.
 \\Therefore,  $\varphi(A^1)=0$. By $\varphi\{x_1, x_2, x_3\}=\{\varphi(x_1), x_2,  x_3\}=\{ x_1,  x_2, \varphi(x_3)\}=0$, we have  $\varphi(A)\subseteq Z(A)$, it follows $\Gamma(A)\cap Der(A)\subseteq Der_C(A)$.

It is clear that if $\varphi\in Der_C(A)$, then $\varphi\in \Gamma(A) $. This implies $Der_C(A)=\Gamma(A)\cap
Der(A)$.
\end{proof}

\begin{theorem}~Let $A$ be a semi-associative 3-algebra. Then  $\forall D \in Der(A), \varphi\in \Gamma(A)$

1) $[D, \varphi]\subseteq \Gamma(A)$,~~
2) $D\varphi\in \Gamma(A)$ if and only if $\varphi D\in Der_C(A)$,

3) $D\varphi\in Der(A)$ if and only if $[D, \varphi]\in Der_C(A)$.

\label{thm:62}
\end{theorem}

{\bf Proof} For any $D\in Der(A), \varphi\in \Gamma(A)$, and $x_1, x_2, x_3\in A$, since

$D\varphi\{x_1, x_2, x_3\}=D\{\varphi(x_1), x_2, x_3\}$

\vspace{2mm}\hspace{-3mm}$=\{D\varphi(x_1), x_2, x_3\}+\{\varphi(x_1), D(x_2), x_3\}+\{\varphi(x_1), x_2, D(x_3)\}$

\vspace{2mm}\hspace{-3mm}$=\{D\varphi(x_1), x_2, x_3\}+\varphi D\{x_1, x_2, x_3\}-\{\varphi D(x_1), x_2, x_3\}$, and

\vspace{2mm}$D\varphi\{x_1, x_2, x_3\}=D\{x_1, x_2, \varphi(x_3)\}$

\vspace{2mm}\hspace{-3mm}$=\{D(x_1), x_2, \varphi(x_3)\}+\{x_1, D(x_2), \varphi(x_3)\}+\{x_1, x_2, D\varphi(x_3)\}$

\vspace{2mm}\hspace{-3mm}$=\varphi D\{x_1, x_2, x_3\}-\{x_1, x_2, \varphi
D(x_3)\}+\{x_1, x_2, D\varphi(x_3)\}$,
 \\ we have $(D\varphi-\varphi
D)\{x_1, x_2, x_3\}=\{(D\varphi-\varphi D)(x_1), x_2, x_3 \}$, and

$(D\varphi-\varphi
D)\{x_1, x_2, x_3\}=\{x_1, x_2, (D\varphi-\varphi D\}(x_3)\}$. Therefore, $[D, \varphi]\in \Gamma(A)$.

Thanks to  Theorem \ref{thm:61}, $\varphi D\in Der A$. If $D\varphi\in  \Gamma(A)$, then by $[D, \varphi]\in \Gamma(A)$, we have   $\varphi D\in \Gamma (A)$.  Therefore,
$\varphi D\in Der(A)\cap \Gamma(A)$.
Conversely, if $\varphi D\in \Gamma (A)$, then by $[D, \varphi]\in \Gamma(A)$, we have $D\varphi \in \Gamma(A)$.
We get 2). The result 3) follows from 1) and 2).

\section{Sub-adjacent  $3$-Lie algebras  and double modules  of  semi-associative 3-algebras}

\subsection{Sub-adjacent  $3$-Lie algebras of semi-associative 3-algebras}

3-Lie algebras have close relationships  with Lie algebras, pre-Lie algebras, associative algebras, commutative associative algebras and et al. Now we discuss
the relation between 3-Lie algebras with  semi-associative 3-algebras.

A $3$-Lie algebra $(L, [ , , ])$ is a vector space $L$  with a linear multiplication $[ , , ]: L\wedge L\wedge L\rightarrow L$ which satisfies that  for all $x_i\in L$, $1\leq i\leq 5,$
\begin{equation}
[[x_1, x_2, x_3], x_4, x_5]=[[x_1, x_4, x_5], x_2, x_3]+[x_1, [ x_2, x_4, x_5], x_3]+[x_1, x_2, [x_3, x_4, x_5]].
\label{eq:41}
\end{equation}

\begin{theorem} Let $A$ be a semi-association $3$-algebra. Then $(A, [ , , ]_c)$ is a $3$-Lie algebra, where for all $x_1, x_2, x_3\in A,$
\begin{equation}
[x_1, x_2, x_3]_c=\{x_1, x_2, x_3\}+\{x_2, x_3, x_1\}+\{x_3, x_1, x_2\}.
 \label{eq:42}
\end{equation}

\end{theorem}

\begin{proof} By \eqref{eq:21}, \eqref{eq:22} and \eqref{eq:23}, the multiplication  $[ , , ]_C$ given by
\eqref{eq:42} is skew-symmetric,  and for all $x_i\in A, 1\leq i\leq 5$,

$[[x_1, x_2, x_3]_c, x_4, x_5]_c-[[x_1, x_4, x_5]_c, x_2, x_3]_c-[x_1, [x_2, x_4, x_5]_c,  x_3]_c$

$-[x_1, x_2, [x_3, x_4, x_5]_c]_c$

\vspace{2mm}\noindent
$=\{\{x_1, x_2, x_3\}, x_4, x_5\}+\{\{x_2, x_3, x_1\}, x_4, x_5\}+\{\{x_3,
x_1, x_2\}, x_4, x_5\}$
$+\{x_4, x_5, \{x_1, x_2, x_3\}\}$

\vspace{2mm}\noindent
$+\{x_4, x_5, \{x_2, x_3, x_1\}\}+\{x_4, x_5,
\{x_3, x_1, x_2\}\}$
$+\{x_5, \{x_1, x_2, x_3\}, x_4\}$
$+\{x_5, \{x_2, x_3, x_1\}, x_4\}$

\vspace{2mm}\noindent$+\{x_5, x_3,
x_1, x_2\}, x_4\}$
$-\{\{x_1, x_4, x_5\}, x_2, x_3\}-\{\{x_4, x_5, x_1\}, x_2, x_3\}$
$-\{\{x_5,
x_1, x_4\}, x_2, x_3\}$

\vspace{2mm}\noindent$-\{x_2, x_3, \{x_1, x_4, x_5\}\}$
$-\{x_2, x_3, \{x_4, x_5, x_1\}\}-\{x_2, x_3,
\{x_5, x_1, x_4\}\}$
$-\{x_3, \{x_1, x_4, x_5\}, x_2\}$

\vspace{2mm}\noindent$-\{x_3, \{x_4, x_5, x_1\}, x_2\}$
$-\{x_3,
\{x_5, x_1, x_4\}, x_2\}$
$-\{x_1, \{x_2, x_4, x_5\}, x_3\}$
$-\{x_1, \{x_4, x_5, x_2\}, x_3\}$

\vspace{2mm}\noindent$-\{x_1,
\{x_5, x_2, x_4\}, x_3\}$
$-\{\{x_2, x_4, x_5\}, x_3, x_1\}-\{\{x_4, x_5, x_2\}, x_3, x_1\}-\{\{x_5,
x_2, x_4\}, x_3, x_1\}$

\vspace{2mm}\noindent
$-\{x_3, x_1, \{x_2, x_4, x_5\}\}-\{x_3, x_1, \{x_4, x_5, x_2\}\}-\{x_3, x_1,
\{x_5, x_2, x_4\}\}$
$-\{x_1, x_2, \{x_3, x_4, x_5\}\}$

\vspace{2mm}\noindent$-\{x_1, x_2, \{x_4, x_5, x_3\}\}-\{x_1, x_2,
\{x_5, x_3, x_4\}\}$
$-\{x_2, \{x_3, x_4, x_5\}, x_1\}-\{x_2, \{x_4, x_5, x_3\}, x_1\}$

\vspace{2mm}\noindent
$-\{x_2,
\{x_5, x_3, x_4\}, x_1\}$
$-\{\{x_3, x_4, x_5\}, x_1, x_2\}-\{\{x_4, x_5, x_3\}, x_1, x_2\}-\{\{x_5,
x_3, x_4\}, x_1, x_2\}$

\vspace{2mm}\noindent
$=\{x_4, x_5, \{x_1, x_2, x_3\}\}+\{x_5, \{x_4, x_2, x_3\}, x_1\}+\{x_4, x_5, \{x_2, x_3, x_1\}\}$
$+\{x_5, \{x_4, x_3, x_1\}, x_2\}$

\vspace{2mm}\noindent$+\{x_4, x_5, \{x_3, x_1, x_2\}\}+\{x_5, \{x_4, x_1, x_2\}, x_3\}$
$-\{x_2, x_3, \{x_1, x_4, x_5\}\}-\{x_3, \{x_2, x_4, x_5\}, x_1\}$

\vspace{2mm}\noindent$-\{x_2, x_3, \{x_4, x_5, x_1\}\}$
$-\{x_3, \{x_2, x_5, x_1\}, x_4\}-\{x_2, x_3, \{x_5, x_1, x_4\}\}-\{x_3, \{x_2, x_1, x_4\}, x_5\}$

\vspace{2mm}\noindent
$-\{x_3, x_1, \{x_2, x_4, x_5\}\}-\{x_1, \{x_3, x_4, x_5\}, x_2\}-\{x_3, x_1, \{x_4, x_5, x_2\}\}$
$-\{x_1, \{x_3, x_5, x_2\}, x_4\}$

\vspace{2mm}\noindent$-\{x_3, x_1, \{x_5, x_2, x_4\}\}-\{x_1, \{x_3, x_2, x_4\}, x_5\}$
$-\{x_1, x_2, \{x_3, x_4, x_5\}\}-\{x_2, \{x_1, x_4, x_5\}, x_3\}$

\vspace{2mm}\noindent$-\{x_1, x_2, \{x_4, x_5, x_3\}\}$
$-\{x_2, \{x_1, x_5, x_3\}, x_4\}-\{x_1, x_2, \{x_5, x_3, x_4\}\}-\{x_2, \{x_1, x_3, x_4\}, x_5\}$
$=0$.

Therefore, the multiplication $[~ , ~ , ~]_c$ satisfies \eqref{eq:41}, and $A_c$ is a 3-Lie algebra. The proof is complete.
\end{proof}

The 3-Lie algebra $(A, [ , , ]_c)$ is called {\it the sub-adjacent $3$-Lie algebra of the semi-association 3-algebra $A$,} and is simply denoted by $A_c$.
 $Der(A_c)$ denotes the derivation algebras of the sub-adjacent $3$-Lie algebra $A_c$, and

\begin{theorem} Let $A$ be a semi-associative $3$-algebra. Then

\vspace{2mm}\noindent
1) derivation algebra $Der(A)$ is a subalgebra of  $Der(A_c)$;

 \vspace{2mm}\noindent 2) for all  $x_i\in A$, $1\leq i\leq 4,$ and $S(x_1, x_2)$, $S(x_3, x_4)\in Der(A)$,
$$ S([x_1, x_2, x_3]_c, x_4)=S(x_2, x_3)S(x_1, x_4)+S(x_1, x_2)S(x_3, x_4)-S(x_1, x_3)S(x_2, x_4); $$

 \vspace{2mm}\noindent
 3) if $I$ is an ideal ( a subalgebra ) of semi-associative algebra $A$, then $I$ is also an ideal ( a subalgebra ) of the $3$-Lie algebra $A_c$.
\label{thm:42}
\end{theorem}

\begin{proof} For any $D\in Der(A)$, and $x, y, z\in A$,

 $D[x, y, z]_c=D(\{x, y, z\}+\{y, z, x\}+\{z, x, y\})$
$=\{Dx, y, z\}+\{x, Dy, z\}+\{x, y, Dz\}$

\vspace{2mm}\hspace{2cm}
$+\{Dy, z, x\}+\{y, Dz, x\}+\{y, z, Dx\}$
$+\{Dz, x, y\}+\{z, Dx, y\}$

\vspace{2mm}\hspace{2.2cm}$+\{z, x, Dy\}$
$=[Dx, y, z]_C+[x, Dy, z]_C+[x, y, Dz]_c,$
\\therefore, $D\in Der(A_c)$. It follows  1).
By the similar discussion to the above, we get  2) and 3).
\end{proof}

\begin{theorem} Let $A$ be a semi-associative $3$-algebra. Then $L(A), R(A)\subseteq $ $Der(A_c)$.

\label{thm:43}
\end{theorem}

\begin{proof}  By \eqref{eq:21}-\eqref{eq:23}, \eqref{eq:32} and \eqref{eq:42}, for any $x, y, u, v, w\in A$,
\begin{align*}
 &L(x, y)[u, v, w]_c
 =L(x, y)(\{u, v, w\}+\{v, w, u\}+\{w, u, v\})\\
 =&\{x, y, \{u, v, w\}\}+\{x, y, \{v, w, u\}\}+\{x, y, \{w, u, v\}\}\\
 =&\{u, \{v, x, y\}, w\}+\{x, \{y, v, w\}, u\}+\{x, \{y, w, u\}, v\}\\
 =&\{w, \{v, x, y\}, u\}+\{u, \{w, x, y\}, v\}+\{x, \{y, v, w\}, u\}+\{x, \{y, w, u\}, v\}
 =0,
\end{align*}
\begin{align*}
 &[L(x, y)(u), v, w]_c+[u, L(x, y)(v), w]_c+[u, v, L(x, y)(w)]_c\\
 =&\{\{x, y, u\}, v, w\}+\{v, w, \{x, y, u\}\}+\{w, \{x, y, u\}, v\}\\
 +&\{u, \{x, y, v\}, w\}+\{\{x, y, v\}, w, u\}+\{w, u, \{x, y, v\}\}\\
 +&\{u, v, \{x, y, w\}\}+\{v, \{x, y, w\}, u\}+\{\{x, y, w\}, u, v\}\\
 =&3\{w, \{v, x, y\}, u\}+3\{u, \{w, x, y\}, v\}+3\{x, \{y, v, w\}, u\}+3\{x, \{y, w, u\}, v\} =0.
\end{align*}
Therefore, $L(x, y)[u, v, w]_c=[L(x, y)(u), v, w]_c+[u, L(x, y)(v), w]_c+[u, v, L(x, y)(w)]_c. $

By the similar discussion to the above, we have
$$R(x, y)[u, v, w]_c=[R(x, y)(u), v, w]_c+[u, R(x, y)(v), w]_c+[u, v, R(x, y)(w)]_c. $$
  The proof is complete.
\end{proof}

\subsection{Double modules  of the semi-associative 3-algebras}

In this section, we discuss double modules of semi-associative $3$-algebras.

 A representation (or a module)  $(V, \rho)$ of a $3$-Lie algebra $(G, [, ,])$ is a vector space $V$ and  a linear transformation $\rho :$ $G\wedge G\longrightarrow End(V)$,
 satisfying, $\forall x_1, x_2, x_3, x_4\in G$,
 \begin{equation}
[\rho(x_1, x_2), \rho(x_3, x_4)]=\rho([x_1, x_2, x_3], x_4)+\rho(x_3, [x_1, x_2, x_4]),
\label{eq:43}
\end{equation}
\begin{equation}
\rho([x_1, x_2, x_3], x_4)=\rho(x_1, x_2)\rho(x_3, x_4)+\rho(x_2, x_3)\rho(x_1, x_4)+\rho(x_3, x_1)\rho(x_2, x_4).
\label{eq:44}
\end{equation}

We know that $(V, \rho)$ is a 3-Lie algebra $(G, [ , , ])$-module if and only if $G\dot+ V$ is a 3-Lie algebra in the multiplication $[ , , ]_{\rho}$,
where $\forall x_1, x_2, x_3\in A$ and $v_1, v_2, v_3\in V,$
 \begin{equation}
[x_1+v_1, x_2+v_2, x_3+v_3]_{\rho}=[x_1, x_2, x_3]+\rho(x_1, x_2)v_3+\rho(x_2, x_3)v_1+\rho(x_3, x_1)v_2.
\label{eq:modulelie}
\end{equation}

\begin{defn} Let $A$ be a semi-associative 3-algebra, $V$ be a vector space, and $\phi,\psi$: $A\times A\longrightarrow End(V)$  be linear mappings. If $\phi$ and $\psi$ satisfy the following properties,  $\forall x_1, x_2, x_3, x_4\in A$,
\begin{equation}
\phi(x_1, x_2)=-\phi(x_2, x_1),
\label{eq:51}
\end{equation}
\begin{equation}
\phi(x_1, \{x_2, x_3, x_4\})=\phi(x_1, x_2)\phi (x_3, x_4),
\label{eq:52}
\end{equation}
\begin{equation}
\psi(\{x_1, x_2, x_3\}, x_4)=\psi(x_1, \{x_2, x_3, x_4\})=\psi(x_1, x_4)\psi(x_2, x_3),
\label{eq:53}
\end{equation}
\begin{equation}
\psi(x_1, x_2)\psi(x_3, x_4)=\psi(x_2, x_1)\phi(x_3, x_4)+\psi(x_1, x_3)\phi(x_2, x_4),
\label{eq:54}
\end{equation}
\begin{equation}
\psi(\{x_1, x_2, x_3\}, x_4)=\phi(x_4, \{x_1, x_2, x_3\})+\psi(\{x_4, x_2, x_3\}, x_1),
\label{eq:55}
\end{equation}
\begin{equation}
\psi(x_1, x_2)\psi(x_3, x_4)=\psi(x_1, x_2)\phi(x_3, x_4)=\phi(x_1, x_3)\psi(x_4, x_2)=-\phi(x_1, \{x_3, x_2, x_4\}),
\label{eq:56}
\end{equation}
\\
then $(\phi, \psi, V)$ is called a double representation of $A$, or is simply called a double module of $A$.
\label{defn:51}
\end{defn}

\begin{theorem}~Let $A$ be a semi-associative 3-algebra, and $l,r$: $A\times A\longrightarrow End(V)$  be linear mappings.
Then triple $(\phi, \psi, V)$ is a double module of $A$ if and only if $( A\dot+ V, \{ , , \}_{\phi\psi})$ is a
semi-associative 3-algebra, where $\forall x_i\in A, v_i\in V, i=1, 2, 3$,
 \begin{equation}
 \{x_1+v_1, x_2+v_2, x_3+v_3\}_{\phi\psi}=\{x_1, x_2, x_3\}+\phi(x_1, x_2)v_3-\psi(x_1, x_3)v_2+\psi(x_2, x_3)v_1.
 \label{eq:57}
\end{equation}

\label{thm:51}
\end{theorem}

\begin{proof}  If  $(\phi, \psi, V)$ is a double  mode of the semi-associative algebra $A$, we will prove that $(A\dot+V, \{ , , \}_{\phi\psi})$ is  a semi-associative 3-algebra.

Thanks to \eqref{eq:57},  for all $x_i\in A, v_i\in V, 1\leq i\leq 5$,

$\{x_1+v_1, x_2+v_2, x_3+v_3\}_{\phi\psi}=-\{x_2+v_2, x_1+v_1, x_3+v_3\}_{\phi\psi}$, \eqref{eq:21} holds.

By \eqref{eq:57}, \eqref{eq:52}, \eqref{eq:53} and \eqref{eq:56},

$\{x_1+v_1, \{x_2+v_2, x_3+v_3, x_4+v_4\}_{\phi\psi}, x_5+v_5\}_{\phi\psi}$

\vspace{1mm}\noindent
$=\{x_1, x_2, \{x_3, x_4, x_5\})+\psi(x_2, \{x_3, x_4, x_5\})v_1-\psi(x_1, \{x_3, x_4, x_5\})v_2$

\vspace{1mm}
$+\phi(x_1, x_2)\psi(x_4, x_5)v_3-\phi(x_1, x_2)\psi(x_3, x_5)v_4+\phi(x_1, x_2)\phi(x_3, x_4)v_5$

\vspace{1mm}\noindent
$=\{x_1+v_1, x_2+v_2, \{x_3+v_3, x_4+v_4, x_5+v_5\}_{\phi\psi}\}_{\phi\psi}$,
\eqref{eq:22} holds.

Thanks to \eqref{eq:57}, \eqref{eq:53}- \eqref{eq:55},

$\{x_1+v_1, \{x_2+v_2, x_3+v_3, x_4+v_4\}_{\phi\psi}, x_5+v_5\}_{\phi\psi}$

\vspace{1mm}\noindent
$=\{x_5, \{x_2, x_3, x_4\}, x_1\}+\{x_1, \{x_5, x_3, x_4\}, x_2\}+(\phi(x_5, \{x_2, x_3, x_4\})$

\vspace{1mm} $+\psi(\{x_5, x_3, x_4\}, x_2))v_1+(\phi(x_1, \{x_5, x_3, x_4\})-\psi(x_5, x_1)\psi(x_3, x_4))v_2$

\vspace{1mm}
   $+(\psi(x_5, x_1)\psi(x_2, x_4)+\psi(x_1, x_2)\psi(x_5, x_4))v_3-(\psi(x_5, x_1)\phi(x_2, x_3)$

   \vspace{1mm}
$+\psi(x_1, x_2)\phi(x_5, x_3))v_4+(\psi(\{x_2, x_3, x_4\}, x_1)-\psi(x_1, x_2)\psi(x_3, x_4))v_5$

\vspace{1mm}\noindent
$=\{x_5+v_5, \{x_2+v_2, x_3+v_3, x_4+v_4\}_{\phi\psi}, x_1+v_1\}\}_{\phi\psi}$

\vspace{1mm}
$+\{x_1+v_1, \{x_5+v_5, x_3+v_3, x_4+v_4\}_{\phi\psi}, x_2+v_2\}_{\phi\psi}$,
 \eqref{eq:23} holds.

 Therefore, $(A\dot+V, \{, , \}_{\phi\psi})$ is  a semi-associative 3-algebra.

Conversely, if $(A\dot+V, \{, , \}_{\phi\psi})$ is a semi-associative algebra. Then  by  \eqref{eq:21},  $\forall x_i\in A, v_i\in V$, $1\leq i\leq 3$,
$$\{x_1+v_1, x_2+v_2, x_3+v_3\}_{\phi\psi}=-\{x_2+v_2, x_1+v_1, x_3+v_3\}_{\phi\psi}.$$
Thanks to \eqref{eq:57},

$\{x_1, x_2, x_3+v_3\}_{\phi\psi}=\{x_1, x_2, x_3\}+\phi(x_1, x_2)v_3$,
$\{x_2, x_1, x_3+v_3\}=\{x_2, x_1, x_3\}+\phi(x_2, x_1)v_3$, \\
  therefore, $\phi(x_1, x_2)=-\phi(x_2, x_1) $,  \eqref{eq:51} holds.

By \eqref{eq:22}, and  \eqref{eq:57},

$\{x_1+v_1, \{x_2+v_2, x_3+v_3, x_4+v_4\}_{\phi\psi}, x_5+v_5\}_{\phi\psi}$

\vspace{1mm}\noindent$=\{x_1+v_1, x_2+v_2, \{x_3+v_3, x_4+v_4, x_5+v_5\}_{\phi\psi}\}_{\phi\psi}$,

\vspace{2mm}$\{x_1+v_1, \{x_2+v_2, x_3+v_3, x_4+v_4\}_{\phi\psi}, x_5+v_5\}_{\phi\psi}$

\vspace{1mm}\noindent $=\{x_1, \{x_2, x_3, x_4\}, x_5\}+\psi(\{x_2, x_3, x_4\}, x_5)v_1-\psi(x_1, x_5)\psi(x_3, x_4)v_2$

\vspace{1mm}
$+\psi(x_1, x_5)\psi(x_2, x_4)v_3-\psi(x_1, x_5)\phi(x_2, x_3)v_4+\phi(x_1, \{x_2, x_3, x_4\})v_5$,

\vspace{2mm}$\{x_1+v_1, x_2+v_2, \{x_3+v_3, x_4+v_4, x_5+v_5\}_{\phi\psi}\}_{\phi\psi}$

\vspace{1mm}\noindent
$=\{x_1, x_2, \{x_3, x_4, x_5\}\}+\psi(x_2, \{x_3, x_4, x_5\})v_1-\psi(x_1, \{x_3, x_4, x_5\})v_2$

\vspace{1mm}
$+\phi(x_1, x_2)\psi(x_4, x_5)v_3-\phi(x_1, x_2)\psi(x_3, x_5)v_4+\phi(x_1, x_2)\phi(x_3, x_4)v_5$.

\vspace{1mm}\noindent \
 If we suppose  $v_i\neq 0, v_j=0$, for  $1\leq i\neq j\leq 5$, then we get
\begin{equation}
\psi(\{x_2, x_3, x_4\}, x_5)=\psi(x_2, \{x_3, x_4, x_5\}),
\label{eq:58}
\end{equation}
\begin{equation}
\phi(x_1, \{x_2, x_3, x_4\})=\phi(x_1, x_2)\phi(x_3, x_4),
\label{eq:59}
\end{equation}
\begin{equation}
\psi(x_1, x_5)\psi(x_3, x_4)=\psi(x_1, \{x_3, x_4, x_5\}),
\label{eq:60}
\end{equation}
\begin{equation}
\psi(x_1, x_5)\psi(x_2, x_4)=\phi(x_1, x_2)\psi(x_4, x_5),
\label{eq:61}
\end{equation}
\begin{equation}
\psi(x_1, x_5)\phi(x_2, x_3)=\phi(x_1, x_2)\psi(x_3, x_5).
\label{eq:62}
\end{equation}

Thanks to \eqref{eq:23}, and \eqref{eq:57},

$\{x_5+v_5, \{x_2+v_2, x_3+v_3, x_4+v_4\}_{\phi\psi}, x_1+v_1\}_{\phi\psi}+\{x_1+v_1, \{x_5+v_5, x_3+v_3, x_4+v_4\}_{\phi\psi}, x_2+v_2\}_{\phi\psi}$

\vspace{1mm}\noindent$=\{x_5, \{x_2, x_3, x_4\}, x_1\}+\{x_1, \{x_5, x_3, x_4\}, x_2\}+(\phi(x_5, \{x_2, x_3, x_4\})$$+\psi(\{x_5, x_3, x_4\}, x_2))v_1$

\vspace{1mm}$+(\phi(x_1, \{x_5, x_3, x_4\})-\psi(x_5, x_1)\psi(x_3, x_4))v_2$$+(\psi(x_5, x_1)\psi(x_2, x_4)+\psi(x_1, x_2)\psi(x_5, x_4))v_3$

\vspace{1mm}$-(\psi(x_5, x_1)\phi(x_2, x_3)$$+\psi(x_1, x_2)\phi(x_5, x_3))v_4+(\psi(\{x_2, x_3, x_4\}, x_1)-\psi(x_1, x_2)\psi(x_3, x_4))v_5$, \\
similarly, for $v_i\neq 0, v_j=0$, for $1\leq i\neq j\leq 5$, we have
\begin{equation}
\psi(x_1, x_5)\psi(x_2, x_4)=\psi(x_5, x_1)\psi(x_2, x_4)+\psi(x_1, x_2)\psi(x_5, x_4),
\label{eq:63}
\end{equation}
\begin{equation}
\psi(x_1, x_5)\phi(x_2, x_3)=\psi(x_5, x_1)\phi(x_2, x_3)+\psi(x_1, x_2)\phi(x_5, x_3),
\label{eq:64}
\end{equation}
\begin{equation}
\psi(x_1, x_5)\psi(x_3, x_4)=\psi(x_5, x_1)\psi(x_3, x_4)-\phi(x_1, \{x_5, x_3, x_4\}),
\label{eq:65}
\end{equation}
\begin{equation}
\psi(\{x_2, x_3, x_4\}, x_5)=\phi(x_5, \{x_2, x_3, x_4\})+\psi(\{x_5, x_3, x_4\}, x_2),
\label{eq:66}
\end{equation}
\begin{equation}
\phi(x_1, \{x_2, x_3, x_4\})=\psi(\{x_2, x_3, x_4\}, x_1)-\psi(x_1, x_2)\psi(x_3, x_4).
\label{eq:67}
\end{equation}

Now we analyze the above identities.  \eqref{eq:52} follows  from \eqref{eq:59}; \eqref{eq:55} follows  from \eqref{eq:66}; and it is clear that \eqref{eq:63} is equivalent to \eqref{eq:61}, \eqref{eq:62} and \eqref{eq:64}; and \eqref{eq:67} is equivalent to \eqref{eq:58}, \eqref{eq:60} and \eqref{eq:66}; \eqref{eq:53} follows from  \eqref{eq:58} and \eqref{eq:59};   \eqref{eq:54} follows from \eqref{eq:61}, \eqref{eq:62} and \eqref{eq:64}; and \eqref{eq:56} follows from  \eqref{eq:61}, \eqref{eq:62}, \eqref{eq:64} and \eqref{eq:65}.
Therefore, $(l, r, V)$ is a double  mode of the semi-associative algebra $A$.
The proof is complete.
\end{proof}

Let $\tau :$ $V\otimes V \rightarrow V\otimes  V $ be the exchange mapping, that is,
$$\tau(x_1, x_2)=(x_2, x_1), \forall x_1, x_2\in V. $$

\begin{theorem}~Let $A$ be a semi-associative $3$-algebra,  and $(l, r, V)$ be the double module of $A$. Then $(V, \rho)$ is a 3-Lie algebra $A_c$-module, where
$$\rho: A_c\wedge A_c \rightarrow A_c, \quad \rho=\phi-\psi\tau+\psi.$$
\label{thm:51c}
\end{theorem}

\begin{proof}
 Thanks to Theorem \ref{thm:51},  $A\dot+ V$ is a semi-associative $3$-algebra in the multiplication \eqref{eq:57}.
Therefore, the multiplication of  the subjacent 3-Lie algebra $(A\dot+ V)_c$ of $A\dot+ V$ is  \eqref{eq:42}, that is, $\forall x_1, x_2, x_3 \in A,$ $ v_1, v_2, v_3\in V$,

$[x_1+v_1, x_2+v_2, x_3+v_3]_c$

 \vspace{1mm}\noindent
 $=\{x_1+v_1, x_2+v_2, x_3+v_3\}_{\phi\psi}+\{x_2+v_2, x_3+v_3, x_1+v_1\}_{\phi\psi}+\{x_3+v_3, x_1+v_1, x_2+v_2\}_{\phi\psi}$

 \vspace{1mm}\noindent
 $=\{x_1, x_2, x_3\}+\phi(x_1, x_2)v_3-\psi(x_1, x_3)v_2+\psi(x_2, x_3)v_1+\{x_2, x_3, x_1\}+\psi(x_2, x_3)v_1$

 \vspace{1mm}\noindent
 $-\psi(x_2, x_1)v_3+\psi(x_3, x_1)v_2+\{x_3, x_1, x_2\}+\phi(x_3, x_1)v_2-\psi(x_3, x_2)v_1+\psi(x_1, x_2)v_3$

 \vspace{1mm}\noindent
 $=[x_1, x_2, x_3]_c+(\phi-\psi\tau+\psi)(x_1, x_2)v_3+(\phi-\psi\tau+\psi)(x_2, x_3)v_1+(\phi-\psi\tau+\psi)(x_3, x_1)v_2. $

 Follows from \eqref{eq:modulelie}, $(V, \rho)$ is a  3-Lie algebra $A_c$, where  $\rho=\phi-\psi\tau+\psi$.
\end{proof}

\begin{theorem}
Let $A$ be a semi-associative 3-algebra, and  $L(x, y), R(x, y):A\times A \rightarrow A$ be left and right multiplications defined as  \eqref{eq:32},  $ \forall x, y\in A$. Then $(L, R, A)$ is a double module of $A$.
\label{thm:52}
\end{theorem}

\begin{proof} Apply Lemma \ref{lem:31}, Theorem \ref{thm:51} and Definition \ref{defn:51}.\end{proof}

The  double module  $(L, R, A)$ of the semi-associative 3-algebra $A$ is called {\it the regular representation of $A$}, or {\it the adjoint module of $A$.}

\begin{theorem} Let $A$ be a semi-associative 3-algebra, and  $(l, r, V)$ be  a double module of  $A$.
Then  $(\phi^{*}, \psi^{*}, V^{*})$ is also a double module of $A$,
where $V^*$ is the dual space of $V$, and $\phi^*, \psi^* : A\times A\longrightarrow End(V^*)$ are defined by
$\forall x, y\in A, v\in V, \xi\in V^*$,
\begin{equation}
\langle \phi^*(x, y)(\xi), v\rangle=-\langle \xi, \phi(x, y)(v)\rangle, ~~ \langle \psi^{*}(x, y)(\xi), v\rangle=-\langle \xi, \psi(x, y)(v)\rangle.
\label{eq:68}
\end{equation}
\label{thm:53}
\end{theorem}

\begin{proof} Apply Lemma \ref{lem:31}, Definition \ref{defn:51} and Theorem \ref{thm:51}.\end{proof}

\begin{coro}
Let $A$ be a semi-associative 3-algebra. Then $(L^*, R^*, A^*)$ is a double module of $A$.
\label{coro:51}
\end{coro}

\begin{proof}
The result follows from Theorem \ref{thm:52} and Theorem \ref{thm:53}, directly.

\end{proof}

\subsection{Double extensions of  semi-associative 3-algebras by cocycles}

\begin{defn} Let  $A$ be a semi-associative 3-algebra, if linear mapping $\theta : A\otimes A\otimes A \rightarrow A^{*}$ satisfies that for all $x, y, z, w, u\in A$,
  \begin{equation}
\theta\{x, y, z\}=-\theta\{y, x, z\},
\label{eq:69}
\end{equation}
\begin{equation}
\theta\{x, \{y, z, w\}, u\}=\theta\{x, y, \{z, w, u\}\},
\label{eq:70}
\end{equation}
\begin{equation}
L^{*}(x, y)\theta\{z, w, u\}=-R^{*}(x, u)\theta\{y, z, w\},
\label{eq:71}
\end{equation}
\begin{equation}
\theta\{x, \{y, z, w\}, u\}-\theta\{u, \{y, z, w\}, x\}
\label{eq:72}
\end{equation}

\hspace{3.7cm}$=\theta\{x, \{u, z, w\}, y\}+R^{*}(x, u)\theta\{y, z, w\}$

\hspace{3.7cm}$-R^{*}(u, x)\theta\{y, z, w\}-R^{*}(x, y)\theta\{u, z, w\}$,
\\then $\theta$ is called a cocycle of  the semi-associative 3-algebra $A$, where $L^*, R^*$ are defined by \eqref{eq:68}.
\label{defn:52}
\end{defn}

\begin{theorem} Let  $A$ be a semi-associative 3-algebra, and  $\theta : A\otimes A\otimes A \rightarrow A^{*}$ be  a  cocycle of $A$. Then  $( A\dot+ A^{*}, \{, , \}_{\theta})$ is a semi-associative 3-algebra, where  $\forall  x_i\in A, \xi_i\in A^*, 1\leq i\leq 3,$

\begin{equation} \{x_1+\xi_1, x_2+\xi_2, x_3+\xi_3\}_{\theta}
=\{x_1, x_2, x_3\}+\theta\{x_1, x_2, x_3\}+L^{*}(x_1, x_2)\xi_3
\label{eq:73}
\end{equation}
\hspace{7.2cm}$-R^{*}(x_1, x_3)\xi_2+R^{*}(x_2, x_3)\xi_1. $
\label{thm:54}
\end{theorem}

\begin{proof} Thanks to \eqref{eq:51} - \eqref{eq:56}, $\forall  x_i\in A, \xi_i\in A^*, 1\leq i\leq 3,$

$\{x_1+\xi_1, x_2+\xi_2, x_3+\xi_3\}_{\theta}$$=-\{x_2+\xi_2, x_1+\xi_1, x_3+\xi_3\}_{\theta}$,  \eqref{eq:21} holds.

By Definition \ref{defn:52},

\vspace{2mm}$\{x_1+\xi_1, \{x_2+h_2, x_3+h_3, x_4+\xi_4\}_{\theta}, x_5+\xi_5\}_{\theta}$

\vspace{1mm}\noindent
$=\{x_1, x_2, \{x_3, x_4, x_5\}\}+\theta\{x_1, x_2, \{x_3, x_4, x_5\}\}+L^{*}(x_1, x_2)L^{*}(x_3, x_4)\xi_5$

\vspace{1mm}$+L^{*}(x_1, x_2)\theta\{x_3, x_4, x_5\}-L^{*}(x_1, x_2)R^{*}(x_3, x_5)\xi_4+L^{*}(x_1, x_2)R^{*}(x_4, x_5)\xi_3$

\vspace{1mm}$-R^{*}(x_1, \{x_3, x_4, x_5\})\xi_2+R^{*}(x_2, \{x_3, x_4, x_5\})\xi_1$

\vspace{1mm}\noindent$=\{x_1+\xi_1, x_2+\xi_2, \{x_3+\xi_3, x_4+\xi_4, x_5+\xi_5\}_{\theta}\}_{\theta}$,  \eqref{eq:22} holds.

Thanks to Theorem \ref{thm:53},

\vspace{2mm}$\{x_1, \{x_2, x_3, x_4\}, x_5\}+\theta\{x_1, \{x_2, x_3, x_4\}, x_5\}+L^{*}(x_1, \{x_2, x_3, x_4\})\xi_5$

\vspace{1mm}$-R^{*}(x_1, x_5)\theta\{x_2, x_3, x_4\}-R^{*}(x_1, x_5)L^{*}(x_2, x_3)\xi_4+R^{*}(x_1, x_5)R^{*}(x_2, x_4)\xi_3$

\vspace{1mm}$-R^{*}(x_1, x_5)R^{*}(x_3, x_4)\xi_2+R^{*}(\{x_2, x_3, x_4\}, x_5)\xi_1$

\vspace{1mm}\noindent$=\{x_5+\xi_5, \{x_2+\xi_2, x_3+\xi_3, x_4+\xi_4\}_{\theta}, x_1+\xi_1\}_{\theta}+\{x_1+\xi_1, \{x_5+\xi_5, x_3+\xi_3, x_4+\xi_4\}_{\theta}, x_2+\xi_2\}_{\theta}$,
\\
we get

\vspace{1mm}\noindent $\{x_1+\xi_1, \{x_2+\xi_2, x_3+\xi_3, x_4+\xi_4\}_{\theta}, x_5+\xi_5\}_{\theta}$

\vspace{1mm}\noindent$=\{x_5+\xi_5, \{x_2+\xi_2, x_3+\xi_3, x_4+\xi_4\}_{\theta}, x_1+\xi_1\}_{\theta}+\{x_1+\xi_1, \{x_5+\xi_5, x_3+\xi_3, x_4+\xi_4\}_{\theta}, x_2+\xi_2\}_{\theta}$,   that is, \eqref{eq:23} holds.
Therefore, $(A\dot+ A^{*}, \{ , , \}_{\theta})$ is a semi-associative 3-algebra.
\end{proof}

The  semi-associative 3-algebra $( A\dot+ A^{*}, \{, , \}_{\theta})$  is called {\it  a double extension of $3$-semi-associative algebra} by $\theta$.

\bibliography{}

\end{document}